\documentclass[letterpaper,10pt,conference]{ieeeconf}
\IEEEoverridecommandlockouts
\usepackage{cite}
\usepackage{amsmath,amssymb,amsfonts}
\usepackage{algorithmic}
\usepackage{graphicx}
\usepackage{textcomp}
\usepackage{xcolor}

\usepackage{balance}
\usepackage{booktabs}
\usepackage{mathtools,mathrsfs,optidef}
\mathtoolsset{showonlyrefs=true}

\newtheorem{theorem}{Theorem}

\newtheorem{lemma}{Lemma}
\newtheorem{corollary}{Corollary}
\newtheorem{remark}{Remark}

\newcommand{\R}{\mathbb{R}}
\newcommand{\Z}{\mathbb{Z}}

\DeclareMathOperator{\diag}{diag}

\DeclareMathOperator{\tr}{tr}

\usepackage{tikz}
\usetikzlibrary{arrows.meta,calc,positioning,shapes}
\tikzset{
  every node/.style = {outer sep=0.12cm, inner sep=0},
  arrow/.style = {-{Triangle[length=0.25cm]}, thick},
  line/.style = {thick},
  block/.style = {rectangle, draw, minimum height=0.8cm, minimum width=2cm, thick, outer sep=0},
  sum/.style = {thick, circle, draw, inner sep=0, minimum size=6 pt, outer sep=0},
  point/.style = {radius=2pt}
}

\newcommand{\graph}{\mathcal{G}}

\title{\LARGE\textbf{Distributed Event-Triggered Consensus Control of Discrete-Time Linear Multi-Agent Systems under LQ Performance Constraints}}
\author{Shumpei~Nishida and Kunihisa~Okano%
  \thanks{This work was supported by JST SPRING Grant Number JPMJSP2101 and JSPS KAKENHI Grant Number 26K07565.}
  \thanks{The authors are with the Graduate School of Science and Engineering,
    Ritsumeikan University,
    Shiga 525-8577 Japan.
    E-mails: {\footnotesize \{\texttt{re0158ff@ed.}, \texttt{kokano@fc.}\}%
    \texttt{ritsumei.ac.jp}}.}
}

\usepackage{fancyhdr}

\begin{document}

\maketitle

\thispagestyle{fancy}
\fancyhf{}
\fancyfoot[C]{%
  \parbox{0.9\textwidth}{%
    \centering
    \fontsize{7}{8}\selectfont
    \copyright\ 2026 IEEE. Personal use of this material is
    permitted. Permission from IEEE must be obtained for all
    other uses, in any current or future media, including
    reprinting/republishing this material for advertising or
    promotional purposes, creating new collective works,
    for resale or redistribution to servers or lists, or
    reuse of any copyrighted component of this work in
    other works.
  }%
}
\renewcommand{\headrulewidth}{0pt}
\renewcommand{\footrulewidth}{0pt}

\pagestyle{empty}

\begin{abstract}
This paper proposes a distributed event-triggered control method that not only guarantees consensus of multi-agent systems but also satisfies a given LQ performance constraint.
Taking the standard distributed control scheme with all-time communication as a baseline, we consider the problem of designing an event-triggered communication rule such that the resulting LQ cost satisfies a performance constraint with respect to the baseline cost while consensus is achieved.
The main difficulty is that the performance requirement is global, whereas triggering decisions are made locally and asynchronously by individual agents, which cannot directly evaluate the global performance degradation.
To address this issue, we decompose allowable degradation across agents and design a triggering rule that uses only locally available information to satisfy the given LQ performance constraint.
For general linear agents on an undirected graph, we derive a sufficient condition that guarantees both consensus and the prescribed performance level.
We also develop a tractable offline design method for the triggering parameters.
A numerical example illustrates the effectiveness of the proposed method.
\end{abstract}

\section{Introduction}

Distributed and cooperative control of multi-agent systems has been extensively studied because of its broad applications, including vehicle formation \cite{ren2008distributed} and power networks \cite{bidram2014distributed}.
In particular, consensus, which aims to achieve agreement among agents through local information exchange, is a fundamental problem and continues to attract significant attention \cite{amirkhani2022consensus}.
However, when communication among agents is implemented over digital networks, frequent information exchange may be impractical because of limited communication resources.
This motivates distributed event-triggered control, in which agents transmit information only when prescribed conditions are satisfied, thereby reducing unnecessary communication while achieving consensus \cite{nowzari2019event}.

Despite this progress, relatively few studies have provided explicit performance guarantees for distributed event-triggered consensus control.
Since consensus only characterizes whether the agents asymptotically reach agreement, it does not directly quantify closed-loop performance.
To address this limitation, several studies have examined the LQ performance of distributed event-triggered consensus control in comparison with time-triggered control \cite{meister2022analysis,meister2024time,antunes2023consistent,meister2026improving}.
Although these studies clarify the performance properties of event-triggered consensus control, their results are limited to specific settings, such as integrator dynamics, complete graphs, and performance criteria that quantify only disagreement among agents.
Therefore, it remains unclear how to design distributed event-triggered communication mechanisms for general linear multi-agent systems with performance guarantees.

In this paper, we study distributed event-triggered consensus control under an LQ performance constraint.
The objective is to achieve consensus while keeping the LQ cost, which penalizes both disagreement and input energy, within a prescribed bound with respect to the all-time communication scheme.
To address the mismatch between the global performance requirement and local asynchronous triggering decisions, we introduce a local performance index and design a distributed triggering rule by decomposing the allowable performance degradation across agents.
We derive sufficient conditions to guarantee consensus and the prescribed performance level, together with a tractable design method for the triggering parameters.
In contrast to existing studies \cite{meister2022analysis,meister2024time,antunes2023consistent,meister2026improving}, the proposed method can be applied to general linear agents and LQ criteria involving both disagreement and input energy.
Moreover, unlike conventional distributed triggering rules~\cite{garcia2014decentralized,hu2015consensus,mishra2023event}, the proposed triggering rule is derived from a performance-oriented analysis.

This paper is organized as follows.
Section~\ref{sec:problem} formulates the consensus problem under the LQ performance constraint.
Section~\ref{sec:standard_consensus} presents preliminary results on consensus control.
Section~\ref{sec:main_results} proposes the distributed event-triggered control method, and Section~\ref{sec:parameter} describes the parameter design procedure.
Section~\ref{sec:numerical} demonstrates a numerical example, and Section~\ref{sec:conclusion} concludes the paper.

Due to space limitations, technical proofs are omitted, and we refer to \cite{nishida2026distributed}.

\paragraph*{Notation}
We denote by $\R$ and $\Z_{\geq 0}$ the sets of real numbers and nonnegative integers, respectively.
The Kronecker product of matrices $A$ and $B$ is written as $A \otimes B$.

\section{Problem Formulation}
\label{sec:problem}

Consider a multi-agent system consisting of $N$ homogeneous agents.
Each agent is modeled as
\begin{equation} \label{eq:mas}
  x_i[k+1] = A x_i[k] + B u_i[k],\quad i \in \{1,\dots,N\},
\end{equation}
where $x_i[k] \in \R^n$ and $u_i[k] \in \R^m$ denote the state and control input, respectively.
We assume that the pair $(A,B)$ is stabilizable, whereas $A$ is not necessarily stable.
The communication topology is represented by a connected weighted undirected graph $\graph$.

We consider the following performance measure:
\begin{equation} \label{eq:lq_cost}
  J(x[0]) = \sum_{k=0}^{\infty} x^\top[k] (L \otimes Q) x[k] + u^\top[k] (I_N \otimes R) u[k],
\end{equation}
where $L$ is the Laplacian matrix of $\graph$, $Q$ and $R$ are positive definite matrices, and
$x[k] \coloneqq [x_1^\top[k]\ \dots\ x_N^\top[k]]^\top$, $u[k] \coloneqq [u_1^\top[k]\ \dots\ u_N^\top[k]]^\top$.
This cost captures both disagreement among the agents and input energy.

In this paper, we study a distributed event-triggered control scheme that guarantees a prescribed level of LQ performance with respect to a scheme with all-time communication.
As a baseline, we consider the standard distributed control law
\begin{equation} \label{eq:nominal_controller}
  u_i[k] = -cF \zeta_i[k],\quad
  \zeta_i[k] \coloneqq \sum_{j \in \mathcal{N}_i} a_{ij} \left(x_i[k] - x_j[k]\right),
\end{equation}
where $a_{ij}$ is the $(i,j)$th entry of the weighted adjacency matrix of $\graph$, $c > 0$ is a coupling gain, and $F \in \R^{m \times n}$ is a feedback gain.
We denote by $J_{\mathrm{all}}(x[0])$ the LQ cost achieved by \eqref{eq:nominal_controller}, where all agents exchange the information required to compute $\zeta_i[k]$ at every time step.

Let $t^i_\ell$ denote the transmission time of agent $i$.
We assume that communication is delay-free and each agent measures its own state every instant.
Using the most recently transmitted information, agent $i$ maintains copies of the predicted states of its neighbors and of itself.
For each $j \in \mathcal{N}_i \cup \{i\}$, these copies are updated according to
\begin{equation} \label{eq:estimation}
  \hat{x}_j[k+1] = A\hat{x}_j[k] + B\hat{u}_j[k],\quad
  \hat{x}_j[t^j_\ell] = x_j[t^j_\ell],
\end{equation}
where $\hat{u}_j[k] = u_j[t^j_\ell]$, $t^j_\ell \leq k < t^j_{\ell+1}$.
Since all such agents use the same transmitted state at each transmission instant and the same predictor dynamics, these copies are identical.
Using \eqref{eq:estimation}, we consider the following distributed event-triggered controller:
\begin{equation} \label{eq:event_triggered_controller}
  u_i[k] = -cF \hat{\zeta}_i[k],\quad
  \hat{\zeta}_i[k] \coloneqq \sum_{j \in \mathcal{N}_i} a_{ij} \left(\hat{x}_i[k] - \hat{x}_j[k]\right).
\end{equation}
We denote by $J_{\mathrm{etc}}(x[0])$ the LQ cost \eqref{eq:lq_cost} under the event-triggered controller \eqref{eq:event_triggered_controller}.

Under the above setup, given a constant $\rho \geq 1$, our objective is to design a distributed event-triggered communication mechanism such that
\begin{equation} \label{eq:lq_performance_constraint}
  J_{\mathrm{etc}}(x[0]) \leq \rho J_{\mathrm{all}}(x[0]),\quad \forall x[0] \in \R^{Nn},
\end{equation}
while ensuring consensus of the multi-agent system, i.e.,
\begin{equation} \label{eq:def_consensus}
  \lim_{k\to\infty} \|x_i[k]-x_j[k]\| = 0,\quad \forall i,j \in \{1,\dots,N\}.
\end{equation}

\begin{remark} \label{rem:baseline_performance}
  Previous studies on event-triggered control have evaluated the LQ performance with respect to a time-triggered control~\cite{antunes2014rollout,khashooei2017output,soleymani2021value,nishida2024feedback}.
  Following these studies, we use the distributed controller with all-time communication as the baseline and impose the relative performance constraint \eqref{eq:lq_performance_constraint}.
\end{remark}

\begin{remark} \label{rem:local_stabilization}
  If local state-feedback controllers are allowed, consensus can be achieved without any communication by independently stabilizing each agent at the origin.
  Hence, we restrict attention to distributed diffusive control based on information exchanged over the graph.
\end{remark}

Before presenting the proposed event-triggered control method, we provide sufficient conditions under which the multi-agent system achieves consensus under the distributed controller \eqref{eq:nominal_controller} and derive an explicit form of the corresponding LQ cost.

\section{Consensus and LQ Cost under All-Time Communication}
\label{sec:standard_consensus}

In this section, we consider consensus control of multi-agent systems under all-time communication, and derive the corresponding LQ cost.

Let $Q_\ell \in \R^{n\times n}$ be a positive semidefinite matrix and $(A,Q_\ell^{1/2})$ be observable.
Since $(A,B)$ is stabilizable, there exists a matrix $P \succ 0$ which is a solution to the following Riccati equation.
\begin{equation} \label{eq:local_riccati_equation}
  P = Q_\ell + A^\top P A - A^\top PB\left(R+B^\top PB\right)^{-1} B^\top PA.
\end{equation}
We design the feedback gain $F$ as
\begin{equation} \label{eq:feedback_gain}
  F = \left(R+B^\top PB\right)^{-1} B^\top PA.
\end{equation}
Moreover, we define
\begin{equation} \label{eq:def_theta}
  \theta \coloneqq \left(\frac{\lambda_{\min}(R)}{\lambda_{\max}(R + B^\top P B)}\right)^{1/2}.
\end{equation}
Then, we obtain a sufficient condition on $c$ to guarantee consensus of the multi-agent system \eqref{eq:mas}.

\begin{lemma}[{\cite{feng2022consensusability}}]
  Suppose that $(A,B)$ is stabilizable and $(A,Q_\ell^{1/2})$ is observable.
  Assume that the undirected graph $\graph$ is connected.
  If $c>0$ is chosen so that
  \begin{equation} \label{eq:condition_c}
    \frac{1}{(1+\theta)\lambda_2(L)} < c < \frac{1}{(1-\theta)\lambda_N(L)},
  \end{equation}
  the multi-agent system \eqref{eq:mas} reaches consensus under the distributed controller \eqref{eq:nominal_controller}.
\end{lemma}

Under the above condition, the LQ cost \eqref{eq:lq_cost} is given by
\begin{equation}
  J_{\mathrm{all}}(x[0])
  =
  \sum_{i=2}^{N} \sum_{k=0}^{\infty}
  \tilde{x}_i^\top[k]W_i\tilde{x}_i[k],
\end{equation}
where $W_i \coloneqq \lambda_i(L)Q + c^2\lambda_i^2(L)F^\top RF$.
Therefore, the corresponding LQ cost admits the following explicit form.

\begin{corollary} \label{cor:lq_cost_alltime}
  Suppose that $F$ is given by \eqref{eq:feedback_gain} and $c$ is designed so as to satisfy \eqref{eq:condition_c}.
  Then, the resulting LQ cost \eqref{eq:lq_cost} is
  \begin{equation} \label{eq:lq_cost_alltime}
    J_{\mathrm{all}}(x[0]) = \sum_{i=2}^{N} \tilde{x}_i^\top[0] P_i \tilde{x}_i[0],
  \end{equation}
  where $P_i,\, i=2,\dots,N$ is the solution to
  \begin{equation} \label{eq:lyapunov_eq_Pi}
    P_i = \left(A - c\lambda_i(L)BF\right)^\top P_i \left(A - c\lambda_i(L)BF\right) + W_i.
  \end{equation}
\end{corollary}

We note that the existence of a unique positive definite solution $P_i \succ 0$ is guaranteed since $A - c\lambda_i(L)BF$ is Schur stable and $W_i \succ 0$ for $i=2,\dots,N$.

\begin{remark} \label{rem:local_cost}
  The LQ cost \eqref{eq:lq_cost} provides a natural performance criterion for consensus control, as it quantifies both the disagreement among agents and the control effort.
  Moreover, distributed consensus controllers of diffusive form have been shown to be optimal for suitably chosen quadratic cost functions \cite{cao2012optimal,zhang2015distributed,feng2022consensusability}.
  In this paper, however, \eqref{eq:lq_cost_alltime} is not necessarily the globally optimal value of \eqref{eq:lq_cost} for the given $Q$ and $R$, but is used as the nominal performance achieved by the standard distributed controller under all-time communication.
\end{remark}

In the following section, we employ the feedback gain $F$ given by \eqref{eq:feedback_gain} and the coupling gain $c$ satisfying \eqref{eq:condition_c}.

\section{Distributed Event-Triggered Consensus under LQ Performance Constraints}
\label{sec:main_results}



In this section, we propose a distributed event-triggered consensus control method under the LQ performance constraint.

First, we introduce several notations used throughout this paper.
Let $e_i[k] \coloneqq \hat{x}_i[k] - x_i[k]$ and $e[k] \coloneqq [e_1^\top[k]\ \dots\ e_N^\top[k]]^\top$.
We consider the following variable:
\begin{equation} \label{eq:apriori_estimation}
  \bar{x}_i[k] = A\hat{x}_i[k-1] + B\hat{u}_i[k-1],
\end{equation}
which serves as the a priori prediction of $x_i[k]$, with the initial state $\bar{x}_i[0] = 0$.
Accordingly, $\hat{x}_i[k]$ is updated by
\begin{equation}
  \hat{x}_i[k] = \begin{cases}
    x_i[k] & \text{if $k = t^i_\ell$ for $\ell \in \Z_{\geq 0}$},\\
    \bar{x}_i[k] & \text{otherwise}.
  \end{cases}
\end{equation}
We also define $\bar{e}_i[k] \coloneqq \bar{x}_i[k] - x_i[k]$ and $\bar{e}[k] \coloneqq [\bar{e}_1^\top[k]\ \dots\ \bar{e}_N^\top[k]]^\top$.
In addition, let us consider
\begin{align}
  \hat\phi_i[k]
  &\coloneqq
  \frac{1}{2}\sum_{j\in\mathcal N_i} a_{ij}
  \left(\hat{x}_i[k]-\hat{x}_j[k]\right)^\top Q \left(\hat{x}_i[k]-\hat{x}_j[k]\right) \\
  &\quad +
  c^2 \hat{\zeta}_i^\top[k] F^\top R F \hat{\zeta}_i[k], \label{eq:def_phi}
\end{align}
which depends only on locally available information to agent $i$.

Suppose that each agent shares its state at time $k=0$, i.e., $t^i_0 = 0$ for all $i \in \{1,\dots,N\}$.
Then, we have $\hat{x}_i[0] = x_i[0]$.
Using \eqref{eq:def_phi}, each agent asynchronously determines transmission instants $\{t^i_{\ell+1}\}_{\ell \in \Z_{\geq 0}}$ as follows:
\begin{equation} \label{eq:triggering_law}
  t^i_{\ell + 1} = \min\{k > t^i_\ell \colon \bar{e}_i^\top[k] \Omega_i \bar{e}_i[k] > \sigma \hat\phi_i[k-1]\},
\end{equation}
where $\Omega_i \succ 0$ and $\sigma > 0$ are design parameters.
Intuitively, agent $i$ triggers when the weighted prediction error becomes too large compared to the local estimated stage cost.
This is because a larger disagreement among agents can tolerate a larger prediction error, whereas near consensus even a small prediction error can have a relatively large impact on the LQ performance.

Note that as a consequence of \eqref{eq:apriori_estimation} and \eqref{eq:triggering_law}, it holds that
\begin{equation}
  e_i[k] = \begin{cases}
    0 & \text{if $k = t^i_\ell$ for $\ell \in \Z_{\geq 0}$}, \\
    \bar{e}_i[k] & \text{otherwise},
  \end{cases}
\end{equation}
which implies that
\begin{equation} \label{eq:inequality_ek_phi}
  e_i^\top[k] \Omega_i e_i[k] \leq \sigma \hat\phi_i[k-1],\quad \forall k\in\Z_{\geq 1},\ i\in\{1,\dots,N\}.
\end{equation}
We also have
\begin{equation} \label{eq:inequality_phi_sum}
  \sum_{i=1} ^N \hat\phi_i[k] = g_{\mathrm{all}}(\hat{x}[k]),
\end{equation}
where $g_{\mathrm{all}}(x) \coloneqq x^\top S x$ and $S \coloneqq L \otimes Q + c^2L^2\otimes F^\top RF$.

For $i \in \{2,\dots,N\}$, define a matrix
\begin{align}
  \Gamma_i
  &\coloneqq c^2\lambda_i(L)^2 F^{\top}B^\top P_i BF \notag\\
  &\quad + \frac{1}{\varepsilon} c^2\lambda_i(L)^2
  F^{\top}B^\top P_i (A - c\lambda_i(L)BF) \notag\\
  &\quad \times W_i^{-1} (A - c\lambda_i(L)BF)^{\top} P_i BF,
\end{align}
and $ \Gamma_U \coloneqq (U \otimes I_n) \diag(0, \Gamma_2, \dots, \Gamma_N) (U^\top \otimes I_n)$.
We also define
\begin{align}
  &\alpha_S \coloneqq \lambda_{\max}\left(\hat\Omega^{-1/2} S \hat\Omega^{-1/2}\right), \\
  &\alpha_{S_u} \coloneqq \lambda_{\max}\left(\hat\Omega^{-1/2} S_u \hat\Omega^{-1/2}\right), \\
  &\alpha_{\Gamma_U} \coloneqq \lambda_{\max}\left(\hat\Omega^{-1/2} \Gamma_U \hat\Omega^{-1/2} \right),
\end{align}
where $\hat\Omega \coloneqq \diag(\Omega_1,\dots,\Omega_N)$ and $S_u \coloneqq c^2L^2\otimes F^\top RF$.

The following lemmas hold, which are used to establish the main results.

\begin{lemma} \label{lem:error_sum_bound}
  Suppose that $\sigma$ satisfies
  \begin{equation} \label{eq:sigma_condition1}
    0 < \sigma < \frac{1}{\alpha_S},
  \end{equation}
  and choose a constant $\eta > 0$ such that
  \begin{equation} \label{eq:eta_condition}
    \eta > \frac{\sigma \alpha_S}{1 - \sigma \alpha_S}.
  \end{equation}
  Then, it holds that
  \begin{equation}
    \sum_{k=0}^{\infty} e^\top[k] \hat\Omega e[k] \leq \beta \sum_{k=0}^{\infty} g_{\mathrm{all}}(x[k]),
  \end{equation}
  where
  \begin{equation} \label{eq:def_beta}
    \beta \coloneqq \frac{\sigma(1 + \eta)}{1 - \sigma(1 + \eta^{-1})\alpha_S}.
  \end{equation}
\end{lemma}

\begin{lemma} \label{lem:sum_gall_bound}
  Let $\beta$ be the constant given in \eqref{eq:def_beta}, and $\varepsilon \in (0,1)$.
  If it holds that $1 - \varepsilon - \alpha_{\Gamma_U} \beta > 0$,
  then
  \begin{equation} \label{eq:def_gamma}
    \gamma \coloneqq \frac{1}{1-\varepsilon-\alpha_{\Gamma_U}\beta}
  \end{equation}
  is well-defined and satisfies
  \begin{equation} \label{eq:gall_sum_bound}
    \sum_{k=0}^\infty g_{\mathrm{all}}(x[k])
    \leq
    \gamma J_{\mathrm{all}}(x[0]).
  \end{equation}
\end{lemma}

The following theorem establishes a sufficient condition to guarantee \eqref{eq:lq_performance_constraint} for a given $\rho > 1$.

\begin{theorem} \label{th:lq_performance_guarantee}
  For $\beta$ and $\gamma$ given in \eqref{eq:def_beta} and \eqref{eq:def_gamma}, respectively, define
  \begin{equation} \label{eq:def_rhostar}
    \hat\rho \coloneqq \left( 1+\delta +\left(1 + \delta^{-1}\right)\alpha_{S_u} \beta \right)\gamma,
  \end{equation}
  where $\delta > 0$.
  Then, for a given $\rho > 1$, if $\hat\rho \leq \rho$ holds, we have
  \begin{equation} \label{eq:lq_performance_guaratee}
    J_{\mathrm{etc}}(x[0]) \leq \rho J_{\mathrm{all}}(x[0]),\quad \forall x[0] \in \R^{Nn}
  \end{equation}
  under \eqref{eq:event_triggered_controller} and \eqref{eq:triggering_law}.
\end{theorem}

\begin{remark}
  Although the performance constraint in \eqref{eq:lq_performance_constraint} is posed for $\rho \geq 1$, Theorem~\ref{th:lq_performance_guarantee} implies that, for any $\rho > 1$, the design parameters can be chosen so that $\hat\rho \leq \rho$ holds.
  Indeed, for any admissible choice of the design parameters, one has $\hat\rho>1$.
  On the other hand, by choosing $\delta > 0$ and $\varepsilon \in (0,1)$ sufficiently small, and then taking $\sigma > 0$ sufficiently small, we can make $\hat\rho$ arbitrarily close to $1$.
\end{remark}

As a consequence of Theorem~\ref{th:lq_performance_guarantee}, the multi-agent system achieves consensus while satisfying the performance constraint \eqref{eq:lq_performance_constraint}.

\begin{theorem} \label{th:consensus}
  Let $\rho > 1$ be given.
  Suppose that $\{\Omega_i\}_{i=1}^N$, $\sigma$, $\delta$, $\varepsilon$, and $\eta$
  are chosen so that $\hat\rho \leq \rho$ holds.
  Then, under \eqref{eq:event_triggered_controller} and \eqref{eq:triggering_law}, the multi-agent system \eqref{eq:mas} reaches consensus.
\end{theorem}

In the next section, we provide a design method for $\{\Omega_i\}_{i=1}^N$, $\sigma$, $\delta$, $\varepsilon$, and $\eta$ to satisfy $\hat\rho \leq \rho$ for a fixed $\rho > 1$.

\section{Parameter Design}
\label{sec:parameter}

For a given performance level $\rho > 1$, we design $\{\Omega_i\}_{i=1}^N$, $\eta$, $\delta$, $\sigma$, and $\varepsilon$ such that $\hat{\rho} \le \rho$.
Although a sufficiently small $\sigma$ satisfies this condition, it leads to frequent transmissions.
We therefore seek to make $\sigma$ as large as possible while satisfying the performance constraint.
The following design procedure is carried out offline.

However, increasing $\sigma$ alone does not necessarily reduce the number of transmissions, since triggering depends on the relative scale of $\sigma$ and $\{\Omega_i\}_{i=1}^N$.
To fix this scale, we impose $\sum_{i=1}^N \tr(\Omega_i) = 1$,
and seek to maximize $\sigma$ under this normalization.

Since directly maximizing $\sigma$ over all design parameters is intractable, we design them sequentially.
For a fixed $\varepsilon$, we first design $\{\Omega_i\}_{i=1}^N$.
For the resulting $\{\Omega_i\}_{i=1}^N$ and a given $\sigma$, we then determine $\eta$ and $\delta$ to minimize $\hat{\rho}$.
Next, we maximize $\sigma$, and finally select $\varepsilon$ by grid search.

We first design $\{\Omega_i\}_{i=1}^N$ for a fixed $\varepsilon$.
By the definitions of $\alpha_S$, $\alpha_{\Gamma_U}$, and $\alpha_{S_u}$, for a positive scalar $\kappa$, the inequalities
\begin{equation}
  \alpha_S \leq \kappa,\quad
  \alpha_{S_u} \leq \kappa,\quad
  \alpha_{\Gamma_U} \leq \kappa
\end{equation}
are respectively equivalent to
\begin{equation}
  S \preceq \kappa \hat{\Omega},\quad
  S_u \preceq \kappa \hat{\Omega},\quad
  \Gamma_U \preceq \kappa \hat{\Omega}.
\end{equation}
Hence, for each fixed $\varepsilon$, we consider minimizing the common upper bound $\kappa$:
\begin{mini}|s|
  {\Omega_1,\ldots,\Omega_N}{\kappa}{}{\label{prob:kappa_original}}
  \addConstraint{S \preceq \kappa \hat{\Omega},\ S_u \preceq \kappa \hat{\Omega},\ \Gamma_U \preceq \kappa \hat{\Omega}}
  \addConstraint{\hat{\Omega} = \diag(\Omega_1,\dots,\Omega_N)}
  \addConstraint{\Omega_i \succ 0,\quad i \in \{1,\dots,N\}}
  \addConstraint{\sum_{i=1}^N \tr(\Omega_i) = 1.}
\end{mini}
Problem~\eqref{prob:kappa_original} is not directly tractable because of the bilinear term $\kappa \hat{\Omega}$.
To remove this bilinearity, we introduce $X_i \coloneqq \kappa \Omega_i$ for $i \in \{1,\dots,N\}$
and define $\hat{X} \coloneqq \diag(X_1,\ldots,X_N)$.
Then, $\hat{X} = \kappa \hat{\Omega}$.
Moreover, the normalization $\sum_{i=1}^N \tr(\Omega_i) = 1$ gives
\begin{equation}
  \sum_{i=1}^N \tr(X_i)
  =
  \kappa \sum_{i=1}^N \tr(\Omega_i)
  =
  \kappa.
\end{equation}
Therefore, for each fixed $\varepsilon$, Problem~\eqref{prob:kappa_original} can be transformed into
\begin{mini}|s|
  {X_1,\ldots,X_N}{\sum_{i=1}^N \tr(X_i)}{}{\label{prob:kappa_convex}}
  \addConstraint{S \preceq \hat{X},\ S_u \preceq \hat{X},\ \Gamma_U \preceq \hat{X}}
  \addConstraint{\hat{X} = \operatorname{diag}(X_1,\ldots,X_N)}
  \addConstraint{X_i \succ 0,\quad i \in \{1,\ldots,N\}.}
\end{mini}
Problem~\eqref{prob:kappa_convex} is a semidefinite program and can be solved using standard solvers.
Let $\{X_i^\star(\varepsilon)\}_{i=1}^N$ be an optimal solution to Problem~\eqref{prob:kappa_convex}.
Then the optimal value of Problem~\eqref{prob:kappa_original} is
\begin{equation}
  \kappa^\star(\varepsilon) = \sum_{i=1}^N \tr(X_i^\star(\varepsilon)),
\end{equation}
and an optimal solution is
\begin{equation} \label{eq:Omega_i_star}
  \Omega_i^\star(\varepsilon) = \frac{X_i^\star(\varepsilon)}{\kappa^\star(\varepsilon)},
  \quad
  i \in \{1,\ldots,N\}.
\end{equation}
We denote by $\alpha_S(\varepsilon)$, $\alpha_{S_u}(\varepsilon)$, and $\alpha_{\Gamma_U}(\varepsilon)$ the corresponding values obtained with $\Omega_i^\star(\varepsilon)$.

We next choose $\eta$ for fixed $\varepsilon$ and $\sigma \in (0,1/\alpha_S(\varepsilon))$.
Since $\hat{\rho}$ is nondecreasing in $\beta$ over the feasible region, minimizing $\hat{\rho}$ with respect to $\eta$ reduces to minimizing $\beta$.
Then
  \begin{equation} \label{eq:eta_star}
    \eta = \eta^\star(\sigma; \varepsilon) \coloneqq \frac{\sqrt{\sigma\alpha_S(\varepsilon)}}{1-\sqrt{\sigma\alpha_S(\varepsilon)}},
  \end{equation}
is the unique minimizer of $\beta$ over
  \begin{equation} \label{eq:admissible_eta}
    \eta > \frac{\sigma\alpha_S(\varepsilon)}{1-\sigma\alpha_S(\varepsilon)},
  \end{equation}
and the corresponding minimum value is
  \begin{equation} \label{eq:beta_min}
    \beta_{\min}(\sigma; \varepsilon) = \frac{\sigma}{\left(1-\sqrt{\sigma\alpha_S(\varepsilon)}\right)^2}.
  \end{equation}

We next choose $\delta$ for the same fixed $\varepsilon$ and $\sigma$.
With $\eta=\eta^\star(\sigma;\varepsilon)$, minimizing $\hat{\rho}$ with respect to $\delta$ reduces to minimizing
  \begin{equation} \label{eq:def_f}
    f(\delta) \coloneqq 1+\delta+\left(1+\delta^{-1}\right)\alpha_{S_u}(\varepsilon)\beta_{\min}(\sigma; \varepsilon).
  \end{equation}
Then
  \begin{equation} \label{eq:delta_star}
    \delta = \delta^\star(\sigma; \varepsilon) \coloneqq \sqrt{\alpha_{S_u}(\varepsilon)\beta_{\min}(\sigma; \varepsilon)}
  \end{equation}
is the unique minimizer of $f(\delta)$ over $\delta > 0$, and the corresponding minimum value is
  \begin{equation} \label{eq:f_min}
    f(\delta^\star(\sigma; \varepsilon)) = \left(1+\sqrt{\alpha_{S_u}(\varepsilon)\beta_{\min}(\sigma; \varepsilon)}\right)^2.
  \end{equation}

Substituting \eqref{eq:eta_star} and \eqref{eq:delta_star} into \eqref{eq:def_rhostar}, for a given $\varepsilon$, the sufficient condition in Theorem~\ref{th:lq_performance_guarantee} reduces to $\underline\rho(\sigma; \varepsilon) \le \rho$,
where
\begin{equation}
  \underline\rho(\sigma; \varepsilon)
  =
  \frac{
    \left(1+\sqrt{\alpha_{S_u}(\varepsilon)\beta_{\min}(\sigma; \varepsilon)}\right)^2
  }{
    1-\varepsilon-\alpha_{\Gamma_U}(\varepsilon)\beta_{\min}(\sigma; \varepsilon)
  }.
\end{equation}
Therefore, for each fixed $\varepsilon$, we determine $\sigma$ by solving
\begin{maxi}|s|
  {\sigma}{\sigma}{\label{prob:maxi_sigma}}{}
  \addConstraint{0 < \sigma < \frac{1}{\alpha_S(\varepsilon)}}
  \addConstraint{1-\varepsilon-\alpha_{\Gamma_U}(\varepsilon)\beta_{\min}(\sigma; \varepsilon) > 0}
  \addConstraint{\underline\rho(\sigma; \varepsilon) \le \rho.}
\end{maxi}
Since $\beta_{\min}(\sigma; \varepsilon)$ is strictly increasing in $\sigma$, the feasible set is an interval.
Hence, its largest feasible point can be computed by bisection, and we denote it by $\sigma^\star(\varepsilon)$.

Finally, we determine $\varepsilon$ and the remaining parameters.
Since $\varepsilon$ affects $\underline\rho(\sigma; \varepsilon)$ both directly and through $\Omega_i^\star(\varepsilon)$ and the associated quantities $\alpha_S(\varepsilon)$, $\alpha_{S_u}(\varepsilon)$, and $\alpha_{\Gamma_U}(\varepsilon)$, we select $\varepsilon$ by grid search.
Specifically, we consider a finite set $\mathcal E \subset (0, 1 - 1/\rho)$, where the upper bound follows from the necessary condition $\frac{1}{1-\varepsilon} < \rho$
obtained by taking the limit $\beta_{\min}(\sigma; \varepsilon) \downarrow 0$.
For each $\varepsilon \in \mathcal E$, we solve \eqref{prob:maxi_sigma} to obtain $\sigma^\star(\varepsilon)$ and select
\begin{equation}
  \varepsilon^\star
  \in
  \arg\max_{\varepsilon \in \mathcal E}
  \sigma^\star(\varepsilon),
\end{equation}
and set $\sigma = \sigma^\star(\varepsilon^\star)$, $\Omega_i = \Omega_i^\star(\varepsilon^\star)$, $\eta = \eta^\star(\sigma^\star; \varepsilon^\star)$, and $\delta = \delta^\star(\sigma^\star; \varepsilon^\star)$
according to \eqref{eq:Omega_i_star}, \eqref{eq:eta_star}, and \eqref{eq:delta_star}, respectively.

\begin{remark}
  The proposed parameter design is heuristic, and hence it does not necessarily maximize $\sigma$ globally or minimize the number of transmissions.
\end{remark}

\section{Numerical Example}
\label{sec:numerical}


Consider a group of eight oscillators whose dynamics are given by the following continuous-time linear system:
\begin{equation} \label{eq:sim_oscillator}
  \dot{x}_i(t) =
  \begin{bmatrix}
    0 & 1 \\
    -1 & 0
  \end{bmatrix} x_i(t)
  +
  \begin{bmatrix}
    0\\
    1
  \end{bmatrix} u_i(t),\quad i \in \{1,\dots,8\},
\end{equation}
where $x_i(t) = [x_{i,1}(t)\ x_{i,2}(t)]^\top \in \R^2$ and $u_i(t) \in \R$ denote the state and control input, respectively.
The communication graph is a cycle graph.
Discretizing \eqref{eq:sim_oscillator} with sampling period $0.05$ yields the discrete-time system \eqref{eq:mas}.
The simulation horizon is set to $200$, and
\begin{equation}
  Q = Q_\ell = \begin{bmatrix}
    2 & 0 \\
    0 & 1
  \end{bmatrix},\quad
  R = 1,\quad
  \rho=1.2.
\end{equation}
Applying the parameter design method in Section~\ref{sec:parameter} yields
\begin{align}
  &\varepsilon = 0.0380,\quad
  \sigma = 8.985 \times 10^{-6},\\
  &\Omega_i =
  \begin{bmatrix}
    0.0286 & 0.0372\\
    0.0372 & 0.0964
  \end{bmatrix},\quad i\in\{1,\dots,8\},
\end{align}
and $\underline\rho = 1.1999$, confirming that the condition in Theorem~\ref{th:lq_performance_guarantee} is satisfied.



\begin{figure}[t]
  \centering
  \includegraphics[width=0.98\linewidth]{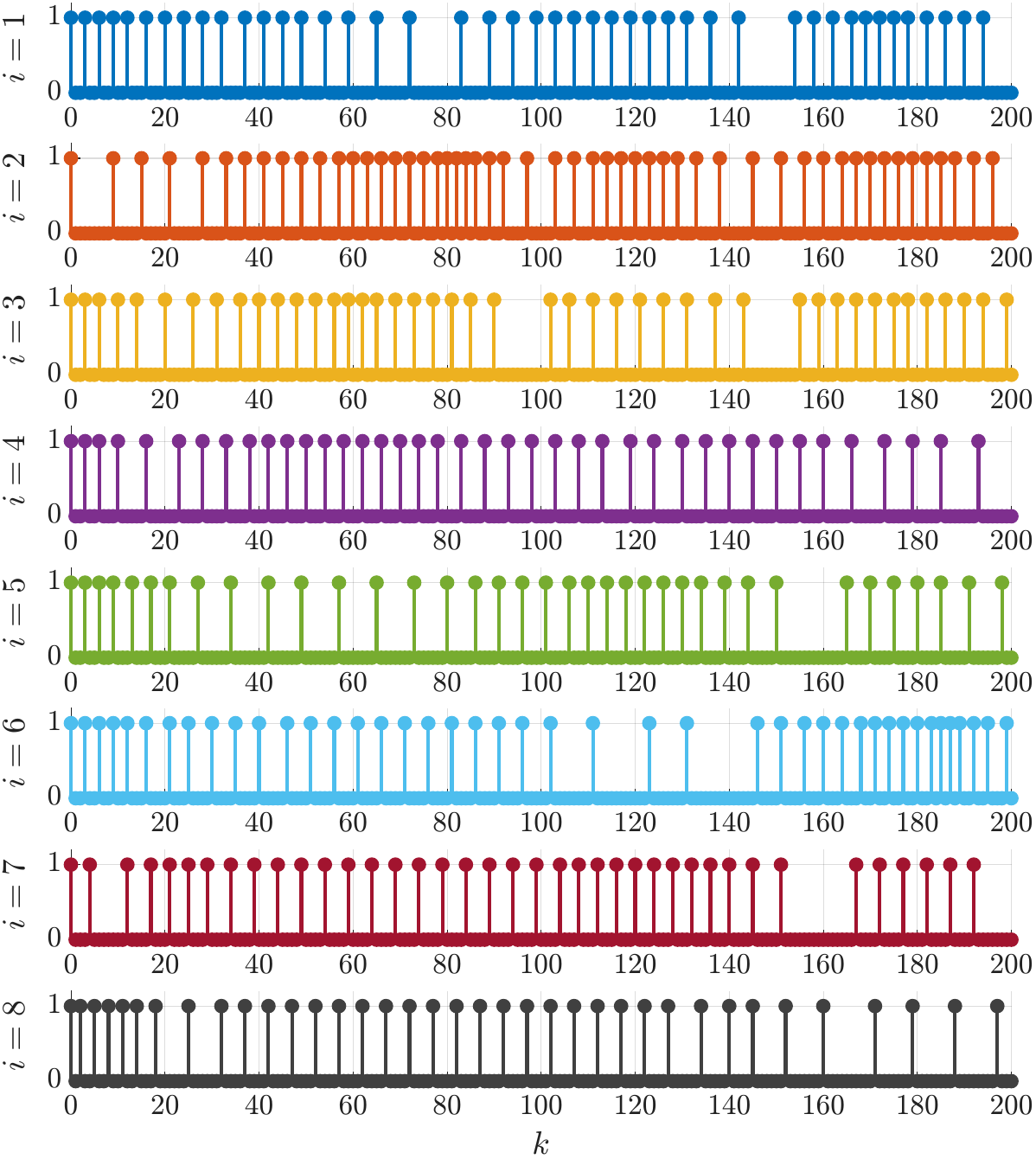}
  \caption{Transmission instants of each agent for the proposed event-triggered method, where a value of $1$ indicates that agent $i$ transmits at that time instant, and a value of $0$ otherwise.}
  \label{fig:transmission}
\end{figure}

\begin{figure}[t]
  \centering
  \includegraphics[width=0.98\linewidth]{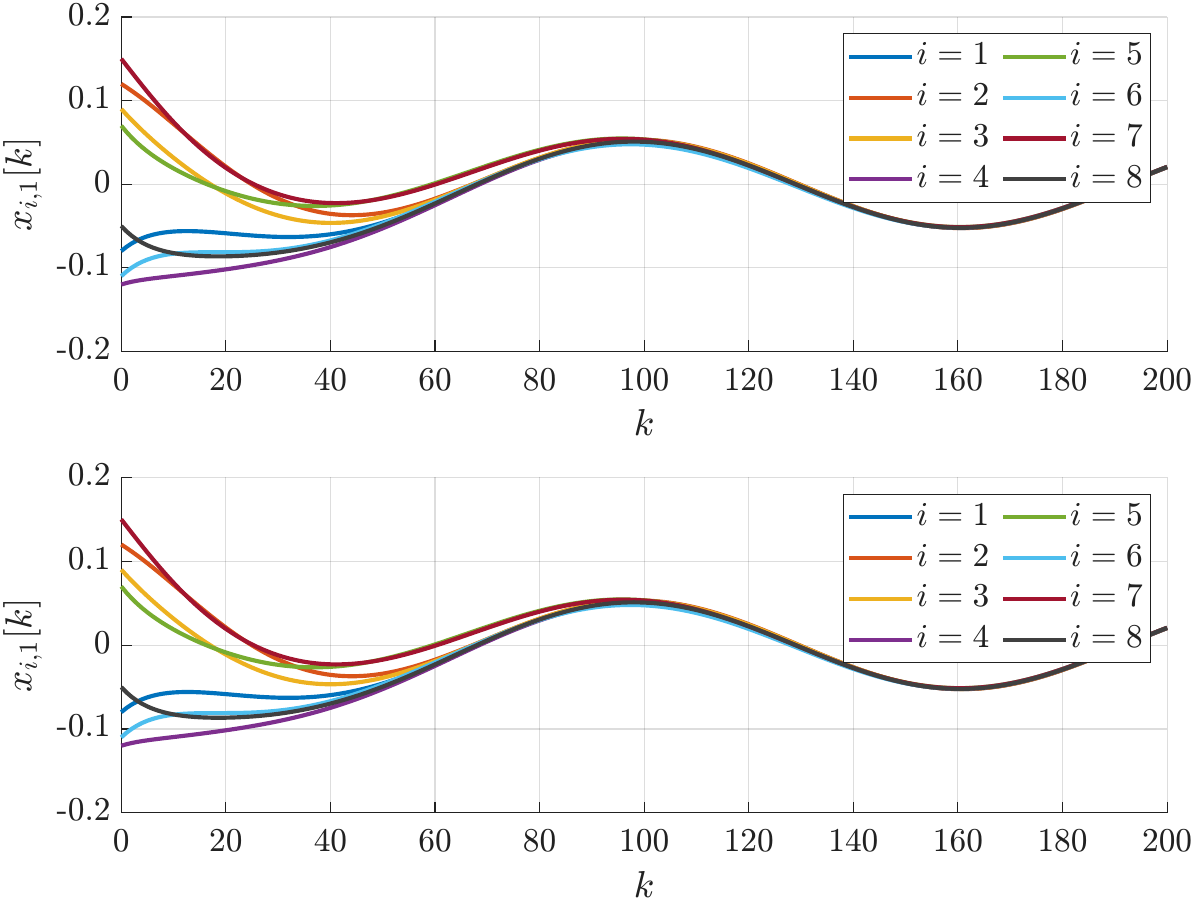}
  \caption{State trajectories $x_{i,1}[k]$ for the all-time communication scheme (top) and the proposed event-triggered method (bottom).}
  \label{fig:state1}
\end{figure}


\begin{figure}[t]
  \centering
  \includegraphics[width=0.98\linewidth]{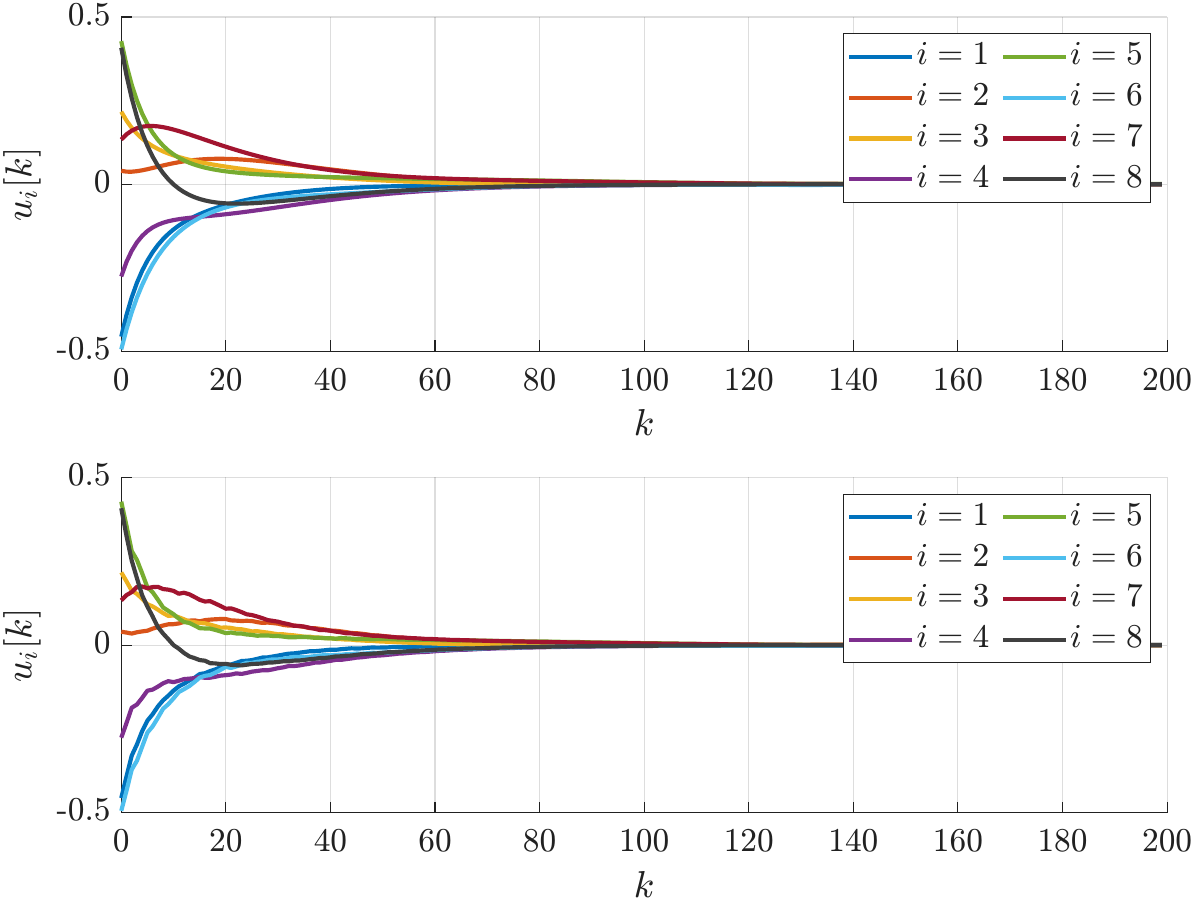}
  \caption{Control inputs for the all-time communication scheme (top) and the proposed event-triggered method (bottom).}
  \label{fig:input}
\end{figure}



\subsection{Simulation Results}

Fig.~\ref{fig:transmission} illustrates the transmission instants of each agent.
Transmissions occur asynchronously at approximately $20.8\%$ of all possible transmission opportunities across all agents.

Fig.~\ref{fig:state1} shows the trajectories of the first state component under the all-time communication scheme and the proposed event-triggered method, while Fig.~\ref{fig:input} shows the corresponding control inputs.
The trajectories under the proposed method converge to consensus, while the transient responses remain close to those under all-time communication despite intermittent transmissions.


\subsection{Discussion}

For this numerical example, the proposed method achieves consensus with trajectories close to those of the all-time communication scheme while using only about $20.8\%$ of all possible transmissions.
The ratio of the costs accumulated over the simulation horizon is $0.9960$, which is well below the prescribed level $\rho=1.2$.
This gap suggests conservatism in the proposed framework, which may arise from both the theoretical analysis and the heuristic parameter design, as discussed below.

First, Theorem~\ref{th:lq_performance_guarantee} provides only a sufficient condition for $J_{\mathrm{etc}}(x[0]) \leq \rho J_{\mathrm{all}}(x[0])$.
Moreover, the global LQ performance constraint is guaranteed through asynchronous triggering decisions based only on locally available information.
Since each agent cannot directly evaluate the exact global performance degradation, the analysis in Section~\ref{sec:main_results} relies on scalar worst-case quantities such as $\alpha_S$, $\alpha_{S_u}$, and $\alpha_{\Gamma_U}$.
Consequently, the resulting condition is tractable but generally conservative.

Second, the parameter design in Section~\ref{sec:parameter} is heuristic rather than globally optimal.
Thus, the obtained $\sigma$ is not necessarily the largest feasible value under the prescribed performance constraint, and maximizing $\sigma$ does not itself guarantee the minimum number of transmissions.

It is also worth noting that $J_{\mathrm{etc}} < J_{\mathrm{all}}$ holds in this simulation despite the lower communication rate.
As stated in Remark~\ref{rem:local_cost}, $J_{\mathrm{all}}$ is not the globally minimal value of the LQ cost \eqref{eq:lq_cost}.
Moreover, the all-time communication scheme and the proposed method induce different closed-loop controllers because they use different information structures.
Therefore, a lower communication rate does not necessarily imply a higher LQ cost.

Overall, the numerical example demonstrates that the proposed framework can substantially reduce communication while maintaining closed-loop performance close to that of all-time communication.
The conservatism of the analytical bound suggests that tighter guarantees may be obtained by more accurately evaluating the effect of local triggering decisions on global performance degradation.

\balance
\section{Conclusion}
\label{sec:conclusion}

This paper studied distributed event-triggered consensus control under an LQ performance constraint.
We proposed a distributed event-triggered control method that guarantees both consensus and the prescribed LQ performance level.
We also developed a tractable parameter design procedure that promotes communication reduction.
Numerical simulations showed that the proposed method reduces transmissions while maintaining transient responses close to those under all-time communication.
Future work includes reducing conservatism in the analysis and parameter design and extending the framework to more general communication settings.

\bibliographystyle{IEEEtran}
\bibliography{myrefs}

\end{document}